\documentclass[journal]{IEEEtran}

\IEEEoverridecommandlockouts
\usepackage{graphicx}
\usepackage{amsmath,amssymb,amsthm}
\usepackage{textcomp}
\usepackage{cite}
\usepackage{mathtools}
\usepackage{subcaption}
\usepackage[nolist]{acronym}
\usepackage[hidelinks]{hyperref}
\usepackage{cleveref}
\usepackage{relsize}
\usepackage{float}
\usepackage[ruled,vlined]{algorithm2e}
\usepackage{tabularx}
\usepackage{booktabs}

\usepackage{pgf,tikz}
\usepackage{pgfplots}
\usetikzlibrary{arrows,positioning} 
\usetikzlibrary{decorations.text}
\usepackage{contour}
\graphicspath{{./graphics/}}

\theoremstyle{definition}

\theoremstyle{definition}
\newtheorem{prop}{Proposition}
\theoremstyle{definition}

\newtheorem{lemma}{Lemma}
\newcommand{\Hd}[1]{H_{#1, \mathrm{d}}}
\newcommand{\Hr}[1]{H_{#1, \mathrm{r}}}

\newcommand*{\la}{\big \lvert}%
\newcommand*{\ra}{\big \rvert}%

\definecolor{blueD}{RGB}{0, 0, 175}
\newenvironment{reply}{}{}
\newcommand{\conffont}{\fontsize{7.5pt}{7.5pt}\selectfont}
\newcommand{\legendfont}{\fontsize{7.5pt}{7.5pt}\selectfont}
\newcommand{\legendfontt}{\fontsize{9pt}{9pt}\selectfont}
\newcommand{\tickfont}{\fontsize{10pt}{10pt}\selectfont}
\newcommand{\labelfont}{\fontsize{11pt}{11pt}\selectfont}
\definecolor{c1}{RGB}{190, 160, 255}
\definecolor{c2}{RGB}{89,51,84}
\definecolor{cgrid}{RGB}{215,215,215}
\definecolor{NOMA}{RGB}{9, 148, 99}
\definecolor{NOMAIRS1}{RGB}{217, 22, 74}
\definecolor{NOMAIRS2}{RGB}{22, 80, 217}

\definecolor{cgrid2}{RGB}{215,215,215}
\definecolor{dark01}{RGB}{50, 50, 50}
\definecolor{greenC1}{RGB}{13, 128, 82}
\definecolor{redC1}{RGB}{209, 6, 60}
\definecolor{blueC1}{RGB}{6, 67, 209}
\definecolor{c3}{RGB}{254, 254, 254}
\tikzset{
	gridc/.style= {dotted, cgrid2},
	linew/.style= {line width=0.65pt},
	linew2/.style= {line width=0.75pt},
	marksz/.style= {mark options={scale=1, fill=white, fill opacity=0.5, solid, line width=0.5}},
	marksz2/.style= {mark options={scale=1.4, fill=white, fill opacity=0.5, solid, rotate=180, line width=0.5}},
	marksz3/.style= {mark options={scale=1, fill=white, fill opacity=0.5, solid, line width=0.5}},
	errormarkSty/.style = {rotate=90, mark size=3pt, solid, line width=0.5pt},
	invisible/.style={opacity=0},
	visible on/.style={alt={#1{}{invisible}}},
	alt/.code args={<#1>#2#3}{%
		\alt<#1>{\pgfkeysalso{#2}}{\pgfkeysalso{#3}} % \pgfkeysalso doesn't change the path
	},
}

\newcommand{\conf}[3]{
	\node[circle, draw=black, name=conf1, anchor=#1, inner sep=1.25pt, minimum size=0pt] at (rel axis cs: #2, #3) {};
	\node[name=conf2, anchor=center, inner sep=3pt, minimum size=0pt] at (conf1) {};
	\node[fill=c3, anchor=west, inner sep=4pt, minimum width=60pt, minimum height=9pt, xshift=-1.5pt, rounded corners=1pt] at (conf2.west) {};
	\draw [draw=dark01](conf2.south) -- (conf2.north); \draw [line width=0.5pt, draw=dark01](conf2.north west) -- (conf2.north east); \draw [line width=0.5pt, draw=dark01](conf2.south west) -- (conf2.south east); \node[right=-2pt of conf2, anchor=west] {\conffont \,$\textrm{95\%}$ confidence}; \node[circle, draw=dark01, anchor=center, inner sep=1.25pt, minimum size=0pt, fill=white, fill opacity=0.5, solid, line width=0.5pt] at (conf1) {};
}

\pgfplotsset{
	axisSetup/.style= {axis x line=bottom, axis y line=left, tick align=inside, axis line style={-, line width=1.25pt, color=dark01}},
	short Legend1/.style={%
		legend image code/.code={
			\draw[##1,linew] plot coordinates {(0pt,0pt) (9pt,0pt)};
			\draw[##1,linew, dashed] plot coordinates {(0pt,2pt) (9pt,2pt)};
			\draw[##1,mark=square, marksz, linew] plot coordinates {(-4pt,1pt)};
		}
	},
	short Legend2/.style={%
		legend image code/.code={
			\draw[##1,linew] plot coordinates {(0pt,0pt) (9pt,0pt)};
			\draw[##1,linew, dashed] plot coordinates {(0pt,2pt) (9pt,2pt)};
			\draw[##1,mark=o, marksz, linew] plot coordinates {(-4pt,1pt)};
		}
	},
	short Legend3/.style={%
		legend image code/.code={
			\draw[##1,linew] plot coordinates {(0pt,0pt) (9pt,0pt)};
			\draw[##1,linew, dashed] plot coordinates {(0pt,2pt) (9pt,2pt)};
			\draw[##1,mark=triangle, linew, marksz2] plot coordinates {(-4pt,1pt)};
		}
	},
	ylabel right/.style={
		after end axis/.append code={
			\node [rotate=90, anchor=north] at (rel axis cs:1,0.5) {#1};
		}   
	}
}

\begin{document}

\author{Bashar Tahir, Stefan Schwarz, and Markus Rupp \\
	
	\thanks{The authors are with the Institute of Telecommunications, TU Wien, Austria (email: bashar.tahir@tuwien.ac.at). Bashar Tahir and Stefan Schwarz are with the Christian Doppler Laboratory for Dependable Wireless Connectivity for the Society in Motion. The financial support by the Austrian Federal Ministry for Digital and Economic Affairs and the National Foundation for Research, Technology and Development is gratefully acknowledged.}

%	Institute of Telecommunications, Technische Universit\"{a}t Wien, Vienna, Austria \\
%	Email: \{bashar.tahir, stefan.schwarz, markus.rupp\}@tuwien.ac.at
}

\title{Analysis of Uplink IRS-Assisted NOMA under Nakagami-\textit{m} Fading via Moments Matching}

\maketitle
\begin{abstract}
This letter investigates the uplink outage performance of \ac{IRS}-assisted \ac{NOMA}. We consider the general case where all users have both direct and reflection links, and all links undergo Nakagami-$m$ fading. We approximate the received powers of the \ac{NOMA} users as Gamma random variables via moments matching. This allows for tractable expressions of the outage under \ac{IC}, while being flexible in modeling various propagation environments. 
%We consider two strategies in which the \ac{IRS} is configured to boost either the first user, or the second one.
 Our analysis shows that under certain conditions, the presence of an \ac{IRS} might degrade the performance of users that have dominant \ac{LOS} to the \ac{BS}, while users dominated by \ac{NLOS} will always benefit from it.
\end{abstract}

%\begin{IEEEkeywords}
%	\ac{IRS}, \ac{NOMA}, Gamma moments matching, outage analysis, interference cancellation. 
%\end{IEEEkeywords}

\IEEEpeerreviewmaketitle

\section{Introduction}
For \ac{B5G} wireless networks, \acfp{IRS} have been identified as a key technology to enhance the spectral- and energy-efficiency at low-cost \cite{Renzo19, Wu19}. Consisting of a large number of reconfigurable nearly-passive elements, those surfaces can alter the propagation of the incident waves to improve the wireless transmission, e.g., boosting the received power, suppressing interference, etc \cite{Wu20}. The combination of \acp{IRS} with \acf{NOMA} has gained interest recently \cite{Fang20, Fu19, Ding20a, desena2020role, Cheng20}. \ac{NOMA} allows multiple users to share the same time-frequency resources, which results in higher spectral-efficiency, lower latency, and/or improved fairness \cite{Ding17}.

Existing works on \ac{IRS}-assisted \ac{NOMA} transmission show promising gains in terms of the outage performance and sum-rate, e.g., \cite{Fang20, Fu19, Ding20a, desena2020role, Cheng20}. However, most of those works consider propagation under certain conditions, such as the weak \ac{NOMA} \ac{UE} being connected to the \acf{BS} only via the \ac{IRS} (no direct link), while the strong \ac{UE} is only served by the direct link to the \ac{BS}, with no contribution from the \ac{IRS}. Another common assumption is Rayleigh fading, which is not a practical model for such systems, since the \ac{IRS} could be deployed at buildings with strong \acf{LOS} to the serving \ac{BS} \cite{Wu19}. Also, in the context of \ac{NOMA} user-pairing, the strong \ac{NOMA} \ac{UE} might have a good \ac{LOS} to the \ac{BS}, and possibly to the surface as well.

In this work, we consider a general two-\ac{UE} \ac{IRS}-\ac{NOMA} uplink in which both \acp{UE} have direct and reflection links to the \ac{BS}, and all the links undergo Nakagami-$m$ fading. \begin{reply}By adjusting the $m$ parameter, we can easily switch between various \ac{LOS} and \acf{NLOS} propagation conditions \cite{Atzeni18}\end{reply}. In order to obtain tractable expressions for the outage under \ac{NOMA} \acf{IC}, the received powers of the \ac{NOMA} \acp{UE} are approximated as Gamma \acp{RV} via moments matching.
Note that the Gamma power approximation has been applied before, e.g., in \cite{Lyu20} to model the received power of \acp{IRS} for an \ac{OMA} setting under Rayleigh fading.
 We consider two strategies in which the \ac{IRS} is either configured to boost the first \ac{UE}, or the second one, and characterize the corresponding channel statistics and outage probability under \ac{IC}. We apply our analysis to an example scenario, and show that under certain conditions, the presence of the \ac{IRS} might degrade the outage performance of the \acp{UE} with dominant \ac{LOS} to the \ac{BS}. 
%On the other hand, \acp{UE} dominated by \ac{NLOS} propagation will always benefit from its presence.
\section{System Model}
\begin{reply}We consider a \ac{NOMA} uplink with two single-antenna \acp{UE}, assisted by an $N$-elements \ac{IRS}. At the \ac{BS} (also single-antenna), the overall received signal from both the direct and reflection links is given by\end{reply}
\begin{align}\label{eq:1}
r = \sum_{i = 1}^{2} \Big(\sqrt{\ell_{h_i}} h_i + \sqrt{\ell_{\mathrm{BS}} \ell_{g_i}} \mathbf{h}_{\mathrm{BS}}^T\mathbf{\Phi}\,\mathbf{g}_i \Big)\sqrt{P_i}\,x_i + w~,
\end{align}
where $h_i \in \mathbb{C}$, $\mathbf{h}_{\mathrm{BS}} \in \mathbb{C}^{N}$, and $\mathbf{g}_i \in \mathbb{C}^{N}$ are the small-scale fading coefficients of the \ac{UE}-\ac{BS}, \ac{BS}-\ac{IRS}, \ac{UE}-\ac{IRS} links, respectively. The parameters $\ell_{h_i}$, $\ell_{\mathrm{BS}}$, and $\ell_{g_i}$ are the corresponding pathlosses, $P_i$ and $x_i$ are the transmit power and signal of the $i^{\mathrm{th}}$-\ac{UE}, and $w$ is the zero-mean Gaussian noise with power $P_w$. The phase-shift matrix $\mathbf{\Phi} \in \mathbb{C}^{N \times N}$ is defined as $\mathbf{\Phi} = \mathrm{diag} \big(e^{j\phi_1}, e^{j\phi_2}, \dots, e^{j\phi_N}\big),$ where $\phi_n$ is the phase-shift applied at the $n^{\textrm{th}}$-element of the \ac{IRS}. Note that the \ac{IRS} term can be written equivalently as
\begin{align}\label{eq:3}
	\mathbf{h}_{\mathrm{BS}}^T\mathbf{\Phi}\,\mathbf{g}_i = \sum_{n = 1}^{N} e^{j\phi_n}  \mathbf{h}_{\mathrm{BS}, n}  \,  \mathbf{g}_{i, n}~,
\end{align}
where $\mathbf{h}_{\mathrm{BS}, n}$ and $\mathbf{g}_{i, n}$ are the $n^{\textrm{th}}$-elements of $\mathbf{h}_{\mathrm{BS}}$ and  $\mathbf{g}_{i}$, respectively. The links are assumed to undergo Nakagami-$m$ fading, i.e., $|h_i| \sim \mathrm{Nakagami} (m_{h_i},\,1)$, $|\mathbf{h}_{\mathrm{BS}, n}| \sim \mathrm{Nakagami} (m_{\mathrm{BS}},\,1)$, and $|\mathbf{g}_{i, n}| \sim \mathrm{Nakagami} (m_{g_i},\,1)$,
%\begin{align}\label{eq:4}
%	\begin{split}
%		|h_i| \sim \mathrm{Nakagami} (m_{h_i}, 1)~, \\
%		|\mathbf{h}_{\mathrm{BS}, n}| \sim \mathrm{Nakagami} (m_{\mathrm{BS}}, 1)~, \\
%		|\mathbf{g}_{i, n}| \sim \mathrm{Nakagami} (m_{g_i}, 1)~,
%	\end{split}
%\end{align}  
where $m_{h_i}$, $m_{\mathrm{BS}}$, and $m_{g_i}$ are the corresponding distribution parameters. On top of being a general fading distribution, Nakagami-$m$ has a Gamma  distributed power, and therefore some of the results we obtain below become exact under certain conditions.  

We consider the case where the \ac{IRS} is configured to boost the received power of either of the \acp{UE}. To maximize the receive power of the $i^{\textrm{th}}$-UE, the phase-shifts are set to $\phi_n = \arg{\big(h_i\big)} - \arg{\big(\mathbf{h}_{\mathrm{BS}, n} \, \mathbf{g}_{i, n}\big)},$
which can be shown by a simple application of the triangular inequality on the received amplitude.
Since the Gamma moments matching is used frequently in this work, we state how it is performed in the following lemma.
%\begin{lemma}[Gamma \ac{RV} second-order moments matching]\label{lemma:1}
\begin{lemma}\label{lemma:1}
	Let $X$ be a non-negative \ac{RV} with first and second moments given by $\mu_X = \mathbb{E}\{X\}$ and $\mu_X^{(2)} = \mathbb{E}\{X^2\}$, respectively. The Gamma \ac{RV} $Y \sim \Gamma(k, \theta)$ with the same first and second moments has shape $k$ and scale $\theta$ parameters
	\begin{align*}
	k = \frac{\mu_X^2}{\mu_X^{(2)} - \mu_X^2}~, \quad \quad \theta = \frac{\mu_X^{(2)} - \mu_X^2}{\mu_X}~.
	\end{align*}		
\end{lemma}
\begin{proof}
\begin{reply}It can be found in statistics books, such as \cite{Florescu14}.\end{reply}
\end{proof}
Additionally, Gamma \acp{RV} satisfy the scaling property, in the sense that if $Y \sim \Gamma(k, \theta)$, then $cY \sim \Gamma(k, c\theta)$. 

\section{Statistics of the Received Power}
Our goal is to obtain expressions that describe the outage probability of the $i^{\textrm{th}}$-\ac{UE}. In the presence of the interference from the other $j^{\textrm{th}}$-\ac{UE}, the \ac{SINR} outage is defined as
\begin{align}\label{eq:15}
p_{\mathrm{out}}^{(i)} = \mathbb{P}\bigg\{ \frac{Z_i P_i}{Z_j P_j + P_w } \leq \epsilon \bigg\} ~,
\end{align}
where $Z_i$ and $Z_j$, as defined below in \eqref{eq:6} and \eqref{eq:12}, are the effective channel powers of the \acp{UE}, and $\epsilon$ is the outage threshold. If the interference is removed via \ac{IC}, then the outage is defined for the \ac{SNR} as
\begin{align}\label{eq:16}
p_{\mathrm{out,\,SNR}}^{(i)} = \mathbb{P}\bigg\{ \frac{Z_i P_i}{P_w } \leq \epsilon \bigg\} ~,
\end{align}
which is simply the \ac{CDF} of $Z_i$ evaluated at $\epsilon P_w / P_i$. In order to evaluate those probabilities, an access to the distributions of $Z_i$ and $Z_j$ is required, which are difficult to characterize, let alone obtaining exact closed-form expressions from them. For that reason, we resort to approximating the received powers as Gamma \acp{RV} via moment matching. On the one hand, the Gamma distribution encompasses many power distributions as special cases, and on the other hand, it allows for tractability when evaluating the outage. To do so, we need access to the moments of $Z_i$ and $Z_j$, for which we first need to characterize their statistics.

Next, and without loss of generality, we assume that the \ac{IRS} is configured to boost \ac{UE}1. In this case, the signal of \ac{UE}1 will be coherently combined, while for \ac{UE}2, and assuming the channels of the two \ac{UE}s are uncorrelated, the combining will be fully random. The other case is simply obtained by a switch of indices. 
%On top of being able to model various propagation conditions, the Gamma distribution also allows for compact expressions when evaluating metrics such as the link \ac{SINR} outage. This is actually one of the reasons why we aimed for having Gamma distributed powers. 
%
%Note that Gamma power approximation has been applied before, e.g., in \cite{Lyu20} to model the receive power of \acp{IRS} for an \ac{OMA} setting under Rayleigh fading.
%% and also in previous works, such as \cite{Heath13}
%
%The \ac{SINR} outage of the $i^{\textrm{th}}$-\ac{UE} is defined as
%\begin{align}\label{eq:15}
%p_{\mathrm{out}}^{(i)} = \mathbb{P}\bigg\{ \frac{Z_i P_i}{Z_j P_j + P_w } \leq \epsilon \bigg\} ~,
%\end{align}
%where $P_w$ is the noise power and $\epsilon$ is the outage threshold. If the interference is removed via \ac{IC}, then the outage is defined for the \ac{SNR} as
%\begin{align}\label{eq:16}
%p_{\mathrm{out, SNR}}^{(i)} = \mathbb{P}\bigg\{ \frac{Z_i P_i}{P_w } \leq \epsilon \bigg\} ~,
%\end{align}
%which is the Gamma \ac{CDF} of $Z_i$ evaluated at $\epsilon P_w / P_i$.

%\vspace{-6mm}
\subsection{Statistics of the Coherently Combined UE}
Since the \ac{IRS} is configured for \ac{UE}1, its signal will be coherently combined, and therefore its effective channel power is given by
\begin{align}\label{eq:6}
Z_1&= \Big(\sqrt{\ell_{h_1}}  \la h_1 \ra + \sqrt{\ell_{\mathrm{BS}} \ell_{g_1}} \sum_{n = 1}^{N} \la \mathbf{h}_{\mathrm{BS}, n} \ra \, \la \mathbf{g}_{1, n} \ra \Big)^2~.
\end{align}
Due to the coherent combining, all the fading terms are positive in-phase aligned, and the second term constitutes a sum of identical unit-power double-Nakagami \acp{RV}. By the causal form of the \ac{CLT} \cite{Papoulis62}, we can approximate the sum of positive \acp{RV} by a Gamma \ac{RV}. This is given by the following lemma (we found similar approximation for Rayleigh fading in \cite{Atapattu20}).

\begin{lemma}\label{prop:1}
	Let $S_1 = \sum_{n = 1}^{N} \la \mathbf{h}_{\mathrm{BS}, n} \ra \, \la \mathbf{g}_{1, n} \ra$, then the distribution of $S_1$ can be approximated as
	\begin{align*}
	S_1 \stackrel{\mathrm{approx}}{\sim}  \Gamma\Bigg(N\frac{\mu_1^2}{1 - \mu_1^2}\,, \, \frac{1 - \mu_1^2}{\mu_1} \Bigg)~,
	\end{align*}
\end{lemma}
with 
\begin{align*}
	\mu_1 =  \frac{\Gamma(m_{\mathrm{BS}} + \frac{1}{2}) \Gamma(m_{g_1} + \frac{1}{2})}{\Gamma(m_{\mathrm{BS}})\Gamma(m_{g_1}) (m_{\mathrm{BS}}\, m_{g_1})^{1/2}}~,
\end{align*}
where $\Gamma(.)$ is the Gamma function.
\begin{proof}
	\begin{reply}Assuming $N$ is large enough, we apply the causal form of the \ac{CLT} and approximate the sum by a Gamma \ac{RV} via \Cref{lemma:1}\end{reply}. For that, we need the first and second moments of the sum. Note that the denominator of $k$ and the numerator of $\theta$ in \Cref{lemma:1} are the variance, which is easier to calculate here. The mean and variance of the sum under unit-power i.i.d. conditions are given by
	\begin{align*}
		\mu_{S_1}^{\vphantom{(2)}} &= \sum_{n = 1}^{N} \mathbb{E}\big\{| \mathbf{h}_{\mathrm{BS}, n} | \, | \mathbf{g}_{1, n} |\big\} = N \mu_1~, \\
		\mu_{S_1}^{(2)} - \mu_{S_1}^2 &= \sum_{n = 1}^{N} \mathrm{Var}\big\{| \mathbf{h}_{\mathrm{BS}, n} | \, | \mathbf{g}_{1, n} |\big\} = N (1 - \mu_1^2)~,
	\end{align*}
	with
	\begin{align*}
		\mu_1 = \mathbb{E}\big\{| \mathbf{h}_{\mathrm{BS}, n} | \, | \mathbf{g}_{1, n} |\big\} = \mathbb{E}\big\{| \mathbf{h}_{\mathrm{BS}, n} |\big\}\mathbb{E}\big\{| \mathbf{g}_{1, n} |\big\}
	\end{align*}
	being the product of the mean of two independent Nakagami \acp{RV}. Substituting the values, we arrive at the final result.
\end{proof}
\begin{reply}The quality of this Gamma approximation improves with the number of IRS elements, as known from the \ac{CLT}. To get a feeling for the approximation, \Cref{fig:f1} shows the density in $\log$-scale for $N = 4$\end{reply}. As can been seen, the approximation holds very well even in the case of only four elements.

\begin{figure}[t]
	\centering
	\resizebox{0.855\linewidth}{!}{%
		\pgfplotsset{width=240pt, height=120, compat = 1.9}
		\begin{tikzpicture}
	\begin{semilogyaxis}[
	xlabel={$S_1$},
	ylabel={Density},
	label style={font=\labelfont},
	ylabel shift = -1mm,	
	ymin=1e-6, ymax=1,
	xmin=0, xmax=10,
	xtick={0,1,2,3,4,5,6,7,8,9,10},
	ytick={1e-6,1e-4,1e-2,1e0},		
	ticklabel style = {font=\tickfont},
	ymajorgrids=true,
	xmajorgrids=true,
%	yminorgrids=true,
%	xminorgrids=true,
	major x grid style={solid, cgrid},
	major y grid style={solid, cgrid},
%	minor x grid style={solid, cgrid},
%	minor y grid style={solid, cgrid},
	legend style={font=\legendfontt, at={(0.5,0.05)},anchor=south},
	legend cell align=left,
	]
	\addplot+[color=c1, mark=none, line width=2pt] table [x=x1, y=y1, col sep=comma] {graphics/results/Fig1.csv};
	\addlegendentry{Empirical}
	\addplot+[color=c2, mark=none, line width=1.25pt, dashed, dash phase=0, dash pattern=on 5pt off
	5pt] table [x=x2, y=y2, col sep=comma] {graphics/results/Fig1.csv};
	\addlegendentry{Gamma approx.}
	\end{semilogyaxis}
\end{tikzpicture}
	}
	\vspace{-1mm}
	\caption{Density of $S_1$ for $N = 4$, $m_{\mathrm{BS}} = 3$, and $m_{g_1} = 1$.}
	\label{fig:f1}
	\vspace{-2mm}
\end{figure}
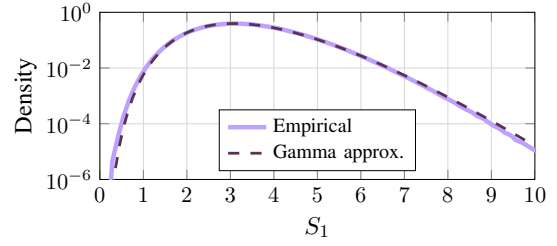

Let $\Hd1 = \sqrt{\ell_{h_1}}  \la h_1 \ra$ and $\Hr1 =\sqrt{\ell_{\mathrm{BS}} \ell_{g_1}}  S_1$ be the pathloss-scaled fading coefficients of the direct and reflection links, the channel power of \ac{UE}1 can now be written as
\begin{align}\label{eq:9}
Z_1 &= (\Hd1 + \Hr1)^2~.
\end{align}
Using the scaling properties, the terms are distributed as
\begin{align}
\Hd1 &\sim \mathrm{Nakagami}\Big(m_{h_1}, \ell_{h_1}\Big)~, \\ 
\Hr1 &\sim \Gamma\Big(Nk_{S_1}, \sqrt{\ell_{\mathrm{BS}} \ell_{g_1}} \theta_{S_1} \Big)~,
\end{align}
where $k_{S_1} = \mu_1^2 / (1 - \mu_1^2)$ and $\theta_{S_1} = (1 - \mu_1^2) / \mu_1$.
Finally, we approximate $Z_1$ by our originally intended Gamma \ac{RV}. Again, we need access to the first and second moments. Those are given under independence by the following lemma.
\begin{lemma}\label{lemma:2}
%The first two moments of \ac{UE}1 receive power (coherent combining) are given by
%\begin{align*}
%\mathbb{E}\{(H_1 + S_1 )^2\} &= \mathbb{E}\{H_1^2\} + \mathbb{E}\{S_1^2\} + 2\mathbb{E}\{H_1\}\mathbb{E}\{S_1\}~, \\
%\mathbb{E}\{(H_1 + S_1 )^4\} &= \mathbb{E}\{H_1^4\} + \mathbb{E}\{S_1^4\} +  6\,\mathbb{E}\{H_1^2\}\mathbb{E}\{S_1^2\}\\ & ~~~ + 4\,\mathbb{E}\{H_1^3\}\mathbb{E}\{S_1\} + 4\,\mathbb{E}\{H_1\}\mathbb{E}\{S_1^3\}~.
%\end{align*}
%where 
%\begin{align*}
%\mathbb{E}\{H_1^d\} &= \frac{\Gamma(m_{h_1} + \frac{d}{2})}{\Gamma(m_{h_1}) (m_{h_1} / \ell_{h_1})^{d/2}}~, \\  
%\mathbb{E}\{S_1^d\} &= \frac{\Gamma(k_s + d) (\sqrt{\ell_{\mathrm{BS}} \ell_{g_1} \vphantom{p_{h_1}} } \theta_s)^d }{\Gamma(k_s)}~.
%\end{align*}

The first two moments of \ac{UE}1 channel power (under coherent combining) are given by
\begin{align*}
\mu_{Z_1}^{\vphantom{(2)}} &= \mu_{\Hd1}^{(2)} + \mu_{\Hr1}^{(2)} + 2\mu_{\Hd1}^{\vphantom{(2)}}\,\mu_{\Hr1}^{\vphantom{(2)}}~, \\
\mu_{Z_1}^{(2)} &= \mu_{\Hd1}^{(4)} + \mu_{\Hr1}^{(4)} +  6\,\mu_{\Hd1}^{(2)}\,\mu_{\Hr1}^{(2)} \\ & \quad\quad+ 4\,\mu_{\Hd1}^{(3)}\,\mu_{\Hr1}^{\vphantom{(2)}} + 4\,\mu_{\Hd1}^{\vphantom{(2)}}\,\mu_{\Hr1}^{(3)}~,
\end{align*}
where 
\begin{align*}
\mu_{\Hd1}^{(p)}  &= \frac{\Gamma(m_{h_1} + \frac{p}{2})}{\Gamma(m_{h_1}) (m_{h_1} / \ell_{h_1})^{p/2}}~, \\  
\mu_{\Hr1}^{(p)}  &= \frac{\Gamma(Nk_{S_1}  + p) (\sqrt{\ell_{\mathrm{BS}} \ell_{g_1} \vphantom{p_{h_1}} } \theta_{S_1} )^p }{\Gamma(Nk_{S_1} )}~.
\end{align*}
\end{lemma}
\begin{proof}
	Proof follows directly by expanding \eqref{eq:9} and substituting the moments of Nakagami and Gamma \acp{RV}.
\end{proof}
\begin{reply}Evaluating $\mu_{\Hr1}^{(2)}$, we find that it is given by $\mu_{\Hr1}^{(2)} =\ell_{\mathrm{BS}} \ell_{g_1} \theta_{S_1}^2 (N^2 k_{S_1}^2 + N k_{S_1})$. This shows a quadratic improvement of the mean received power of $Z_1$ with the number of \ac{IRS} elements $N$. However, notice the presence of the composite pathloss of the channel $\ell_{\mathrm{BS}} \ell_{g_1}$, meaning that in order for the promised gains to be realized, $N$ should be large enough such that the \ac{IRS} can overcome the pathloss\end{reply}. After scaling with $P_1$, the \ac{UE}1 receive power follows the distribution
\begin{align}
	Z_1P_1 \stackrel{\mathrm{approx}}{\sim}  \Gamma\big(k_1, P_1\theta_1\big)~,
\end{align}
where $k_1$ and $\theta_1$ are the Gamma \ac{RV} parameters matched to the moments in \Cref{lemma:2}.

\vspace{-4mm}
\subsection{Statistics of the Randomly Combined UE}
Assuming the channels of the users are uncorrelated, the combining will appear random for \ac{UE}2. In that case, the effective channel power of \ac{UE}2 is given by
\begin{align}\label{eq:12}
Z_2 &= \Big\lvert ~ \sqrt{\ell_{h_2}} h_2  + \sqrt{\ell_{\mathrm{BS}} \ell_{g_2} \vphantom{p_{h_2}} } \sum_{n = 1}^{N} e^{j\phi_n} \mathbf{h}_{\mathrm{BS}, n} \, \mathbf{g}_{2, n} ~ \Big \rvert^2~.
\end{align}
Compared to \eqref{eq:6}, the sum term consists of out-of-phase complex-valued coefficients. Similarly to the previous subsection, we attempt to fit the sum by a simple distribution; namely, a complex-Gaussian through the conventional \ac{CLT}.
\begin{lemma}\label{lemma:3}
Let $S_2 = \sum_{n = 1}^{N} e^{j\phi_n} \mathbf{h}_{\mathrm{BS}, n} \, \mathbf{g}_{2, n}$, then the distribution of $S_2$ can be approximated as
\begin{align*}
S_2 \stackrel{\mathrm{approx}}{\sim} \mathcal{CN}(0, N)~.
\end{align*}
\end{lemma}
\begin{proof}
\begin{reply}Proof follows by application of the \ac{CLT} on the sum of complex unit-power i.i.d. \acp{RV}. As shown in \cite{Ding20b}, even under correlation of the real and imaginary parts, the \ac{CLT} still provides a good approximation.\end{reply}
\end{proof}

To see how good such an approximation is, we compare the magnitude of the sum with a Rayleigh fit (magnitude of Gaussian). This is shown in \Cref{fig:f2} for $N = 4$. We see that it does provide a good fit; however, it is not as good (at the tails) compared to the Gamma approximation in the case before.

\begin{figure}[t]
	\centering
	\resizebox{0.855\linewidth}{!}{%
		\pgfplotsset{width=240pt, height=120, compat = 1.9}
		\begin{tikzpicture}
	\begin{semilogyaxis}[
	 	xlabel={$|S_2|$},
		ylabel={Density},
		label style={font=\labelfont},
		ylabel shift = -1mm,	
		ymin=1e-6, ymax=1,
		xmin=0, xmax=8,
		xtick={0,1,2,3,4,5,6,7,8},
		ytick={1e-6,1e-4,1e-2,1e0},	
		ticklabel style = {font=\tickfont},
		ymajorgrids=true,
		xmajorgrids=true,
		major x grid style={solid, cgrid},
		major y grid style={solid, cgrid},
		legend style={font=\legendfontt, at={(0.5,0.05)},anchor=south},
		legend cell align=left,
	]
	\addplot+[color=c1, mark=none, line width=2pt] table [x=x1, y=y1, col sep=comma] {graphics/results/Fig2.csv};
		\addlegendentry{Empirical}
	\addplot+[color=c2, mark=none, line width=1.25pt, dashed, dash phase=0, dash pattern=on 5pt off 5pt] table [x=x2, y=y2, col sep=comma] {graphics/results/Fig2.csv};
	 \addlegendentry{Rayleigh approx.}
	\end{semilogyaxis}
\end{tikzpicture}
	}
	\vspace{-1mm}
	\caption{Density of $|S_2|$ for $N = 4$, $m_{\mathrm{BS}} = 3$, $m_{g_2} = 1$, and $\phi_n$ being uniformly distributed.}
	\label{fig:f2}
	\vspace{-4mm}
\end{figure}
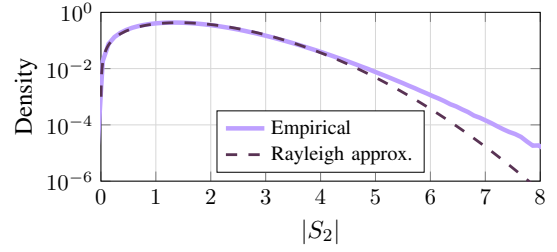

We proceed next in a similar fashion as in the coherent combing case. Let $H_{2, \mathrm{d}} = \sqrt{\ell_{h_2}}  h_2 $ and $H_{2, \mathrm{r}} =\sqrt{\ell_{\mathrm{BS}} \ell_{g_2}} S_2$ be the pathloss-scaled fading coefficients of the direct and reflection links, the channel power of \ac{UE}2 can be written as
\begin{align}\label{eq:13}
Z_2 &= | H_{2, \mathrm{d}} + H_{2, \mathrm{r}} |^2~.
\end{align}
The distribution of the magnitudes is given by
\begin{align}
|H_{2, \mathrm{d}}| &\sim \mathrm{Nakagami}\Big(m_{h_2}, \ell_{h_2}\Big)~, \\ 
|H_{2, \mathrm{r}}| &\sim \mathrm{Nakagami}\Big(1, N \ell_{\mathrm{BS}} \ell_{g_2}\Big)~,
\end{align}
where the fact that Nakagami becomes Rayleigh for $m = 1$ has been applied here to unify notation. Now, we apply the Gamma approximation of the power for $Z_2$. The first and second moments under independence are given by the following lemma.
\begin{lemma}\label{lemma:4}
	The first two moments of \ac{UE}2 channel power (under random combining) are given by
	\begin{align*}
	\mu_{Z_2}^{\vphantom{(2)}} &= \mu_{|\Hd2|}^{(2)} + \mu_{|\Hr2|}^{(2)}~, \\
	\mu_{Z_2}^{(2)} &= \mu_{|\Hd2|}^{(4)} + \mu_{|\Hr2|}^{(4)} +  4\,\mu_{|\Hd2|}^{(2)}\,\mu_{|\Hr2|}^{(2)}~,
	\end{align*}
	where
	\begin{align*}
	\mu_{|\Hd2|}^{(p)}  &= \frac{\Gamma(m_{h_2} + \frac{p}{2})}{\Gamma(m_{h_2}) {(m_{h_2} / \ell_{h_2})}^{p/2}}~, \\  
	\mu_{|\Hr2|}^{(p)}  &= \Gamma\Big(1 + \frac{p}{2}\Big) {(N \ell_{\mathrm{BS}} \ell_{g_2})}^{p/2} ~.
	\end{align*}
\end{lemma}
\begin{proof}
	\begin{reply}For $\mu_{Z_2}^{\vphantom{(2)}}$, the proof is trivial. As for $\mu_{Z_2}^{(2)}$, we have
	\begin{align*}
		\begin{split}
			\mu_{Z_2}^{(2)} = \,& \mathbb{E}\big\{ |\Hd2|^4 + |\Hr2|^4 + 2|\Hd2|^2|\Hr2|^2  \\
			& + 4|\Hd2|^2\Re\{\Hd2^*\Hr2\} + 4|\Hr2|^2\Re\{\Hd2^*\Hr2\} \\
			& + 4 \Re\{\Hd2^*\Hr2\}^2 \big\}\,.
		\end{split}	
	\end{align*}	
	Since  $\Re\{\Hd2^*\Hr2\} = \Re\{\Hd2\}\Re\{\Hr2\} + \Im\{\Hd2\}\Im\{\Hr2\}$, then under the assumptions of independence and zero-mean symmetry, we get $\mathbb{E}\big\{ \Re\{\Hd2^*\Hr2\}\big\} = 0$.
	Assuming the power is equal across the real and imaginary parts, we have
	\begin{align*}
	\begin{split}
		\mathbb{E}\big\{ \Re\{\Hd2^*\Hr2\}^2\big\} =\,& \mathbb{E}\big\{\Re\{\Hd2\}^2\big\} \mathbb{E}\big\{\Re\{\Hr2\}^2\big\} \\ &~~+ \mathbb{E}\big\{\Im\{\Hd2\}^2\big\} \mathbb{E}\big\{\Im\{\Hr2\}^2\big\} \\
		=\,& \mathbb{E}\big\{ |\Hd2|^2\big\} \mathbb{E}\big\{ |\Hr2|^2\big\} / 2\,,
	\end{split}
	\end{align*}
	where the last equation follows from $\mathbb{E}\big\{\Re\{\Hd2\}^2\big\} = \mathbb{E}\big\{\Im\{\Hd2\}^2\big\} = \mathbb{E}\big\{ |\Hd2|^2\big\} / 2$, and similarly for the reflection term. Collecting the terms, and substituting the moments of Nakagami \acp{RV}, we get the final results.\end{reply}
%	Expand \eqref{eq:13} in terms of the complex-conjugate, and then, under the assumption that the phase distribution of the sum term is zero-mean symmetric, all terms involving odd-order moments of $S_2$ will be equal to zero, and thus are not needed. Finally, we substitute the moments of Nakagami \acp{RV}.
\end{proof}
\begin{reply}Evaluating $\mu_{|\Hr2|}^{(2)}$, it is given by $N \ell_{\mathrm{BS}} \ell_{g_2}$. This shows a linear increase of the mean receive power of $Z_2$ with $N$, in contrast to the quadratic increase under the coherent combining case\end{reply}. After scaling with $P_2$, the \ac{UE}2 receive power follows the distribution
\begin{align}
Z_2P_2 \stackrel{\mathrm{approx}}{\sim}  \Gamma\big(k_2, P_2\theta_2\big)~,
\end{align}
where $k_2$ and $\theta_2$ are the Gamma parameters matched to the moments in \Cref{lemma:4}.

\vspace{-1mm}
\section{Interference Cancellation Outage Analysis}
In the following, we calculate the outage probability for the uplink \ac{IRS}-\ac{NOMA} system under \ac{IC}. First, we evaluate the outage probability without \ac{IC}. 
\begin{prop}\label{prop:2}
	
	Let $Z_i P_i\sim  \Gamma(k_i, P_i\theta_i)$ be the received power of the $i^{\textrm{th}}$-\ac{UE}, $Z_j P_j\sim  \Gamma(k_j, P_j\theta_j)$ the received power of the $j^{\textrm{th}}$-\ac{UE}, with $P_w$ being the noise power. \begin{reply}The \ac{IRS}-\ac{NOMA} outage probability without \ac{IC} is given by
	\begin{align*}
	\begin{split}
	p_{\mathrm{out}}^{(i)}  \approx & ~ I\bigg(\frac{\epsilon \hat{\theta}_j}{\hat{\theta}_i + \epsilon \hat{\theta}_j};\,\hat{k}_i,\,\hat{k}_j \bigg)~,
	\end{split}
\end{align*}
where $I(.;.,.)$ is the regularized incomplete beta function\end{reply}, and
\begin{align*}
	\hat{k}_i = k_i~, \quad \quad &\hat{\theta}_i = \theta_iP_i~,\\
	\hat{k}_j = \frac{\big(k_j\theta_jP_j + P_w\big)^2}{k_j (\theta_jP_j)^2}~, \quad \quad &\hat{\theta}_j = \frac{k_j (\theta_jP_j)^2}{k_j\theta_jP_j + P_w}~.
\end{align*}
\end{prop}
\begin{proof}
	\begin{reply}Let $X \sim \Gamma(k_X, \theta_X)$ and $Y \sim \Gamma(k_Y, \theta_Y)$ be two independent Gamma \acp{RV}, then their ratio $R = X / Y$ is known to be beta prime distributed, i.e.,
	\begin{align*}
		R \sim \beta^{'}\big(k_X, k_Y, 1, \theta_X / \theta_Y\big)~,
	\end{align*}
	with its \ac{CDF} given by 
	\begin{align*}
		\mathbb{P}\big\{R \leq \epsilon \big\} = I\bigg(\frac{\epsilon {\theta}_Y}{{\theta}_X + \epsilon {\theta}_Y};\,{k}_X,\,{k}_Y \bigg)~.
	\end{align*}\end{reply}
	However, the denominator in \eqref{eq:15} is not Gamma, due to the presence of the noise term. Therefore, we approximate the interference-plus-noise term by an equivalent Gamma \ac{RV}, again, via moments matching. By doing so, and using the Gamma scaling property, we are arrive at the final results.
\end{proof}

The detection scheme we consider here is parallel, in the sense that \ac{UE}1, \ac{UE}2, or both can be detected correctly at the first iteration and removed from the received signal. Whatever remains can be detected in the second iteration after \ac{IC}. Such formulation allows us to assume an arbitrary cancellation order and save  us the hassle of ordered statistics as would be required under successive \ac{IC}. This is formulated in the following proposition.

\begin{prop}\label{prop:3}
	 The \ac{IRS}-\ac{NOMA} outage probability of the $i^{\textrm{th}}$-\ac{UE} under \ac{IC} is given by
	\begin{align*}
		p_{\mathrm{out,\,IC}}^{(i)} \approx 1 - \min \big(p_{\mathrm{succ}}^{(i)} + p_{\mathrm{succ}}^{(j)} \, p_{\mathrm{succ,\,SNR}}^{(i)} \,,\, p_{\mathrm{succ,\,SNR}}^{(i)} \big)~,
	\end{align*}
%	where $p_{\mathrm{succ}}^{(i)} = 1 - p_{\mathrm{out}}^{(i)}$ and $p_{\mathrm{succ,\,SNR}}^{(i)} = 1 - p_{\mathrm{out,\,SNR}}^{(i)}$ are the success probabilities, with $p_{\mathrm{out,\,SNR}}^{(i)}$ as defined in \eqref{eq:16} is the Gamma \ac{CDF} evaluated at $\epsilon P_w / P_i$.
	where $p_{\mathrm{succ}}^{(i)} = 1 - p_{\mathrm{out}}^{(i)}$ and $p_{\mathrm{succ,\,SNR}}^{(i)} = 1 - p_{\mathrm{out,\,SNR}}^{(i)}$ are the success probabilities, \begin{reply}with $p_{\mathrm{out,\,SNR}}^{(i)}$ as defined in \eqref{eq:16} is the Gamma \ac{CDF} given by
	\begin{align*}
		p_{\mathrm{out,\,SNR}}^{(i)} = \gamma(k_i,\, \epsilon P_w / (\theta_iP_i))~,
	\end{align*}
	with $\gamma(.,\,.)$ being the regularized incomplete Gamma function.\end{reply}
\end{prop}
\begin{proof}
	There are two paths for a successful detection of the $i^{\textrm{th}}$-\ac{UE}: it is detected correctly in the first iteration; or, it is not, but the other \ac{UE} is detected correctly, and after \ac{IC}, the $i^{\textrm{th}}$-\ac{UE} is detected interference-free in the presence of noise only. Following those events, we can approximate the success probability under \ac{IC} as $p_{\mathrm{succ,\,IC}}^{(i)} \approx p_{\mathrm{succ}}^{(i)} + p_{\mathrm{succ}}^{(j)} \, p_{\mathrm{succ,\,SNR}}^{(i)}$. The detection sequence just mentioned is not of fully independent events; hence, the approximation sign. To further improve the approximation, we use the fact that the performance cannot be better than that of the interference-free noise-only case. We get the final results by taking the minimum between this expression and the noise-only case.
\end{proof}

\vspace{-6mm}
\section{Analysis of an Example Scenario}
We assume a scenario where \ac{UE}1 is received at the \ac{BS} with $10$\,dB higher power than \ac{UE}2 through the direct links, and they are assisted by a $32$-elements \ac{IRS}. The strong user (\ac{UE}1) and the \ac{IRS} are assumed to have good \ac{LOS} to the \ac{BS}, while the weak user (\ac{UE}2) experiences close to Rayleigh fading. The two \ac{UE}s are assumed to have a moderate \ac{LOS} to the \ac{IRS}. This is set by adjusting the corresponding $m$ parameters. The scenario parameters are summarized in \Cref{table:1}. 
\begin{table}
	\begin{center}
		\fontsize{8pt}{10pt}\selectfont
		\begin{reply}\begin{tabularx}{0.82\linewidth}{l |l}
			\hline
			\textbf{Parameter} & \textbf{Value} \\ \specialrule{1pt}{0pt}{1pt}
			\#IRS elements & $N = 32$ \\
			Transmit powers & $P_1 = P_2 = P =$ $20$ or $35$\,dBm \\
			Nakagami parameters & \parbox[t]{5cm}{$m_{\mathrm{BS}} = 6$\\
				$m_{h_1} = 4$, $m_{h_2} = 1.1$\\
				$m_{g_1} = m_{g_2} = 2.25$\vspace{1mm}} \\ 
			Pathlosses & \parbox[t]{5cm}{$\ell_{\mathrm{BS}} = -60$\,dB	\\
				$\ell_{h_1} = -110$\,dB, $\ell_{h_2} = -120$\,dB\\
				$\ell_{g_1} = \ell_{g_2} = -60$\,dB\vspace{1mm}}\\
			Noise power & $P_w = -100$\,dBm \\
			\hline
		\end{tabularx}\end{reply}
	\end{center}
	\vspace{-2mm}
	\caption{Scenario parameters.}
	\label{table:1}
	\vspace{-5mm}
\end{table}

We investigate the outage under \ac{IC} (via \Cref{prop:3}) for the two strategies where the \ac{IRS} is either configured to boost \ac{UE}1, which is what our analysis was based on, or, the \ac{IRS} is configured to boost \ac{UE}2, which is obtained by switching indices in the analysis, i.e., \ac{UE}2 is coherently combined, while \ac{UE}1 is randomly combined. As a baseline, we show the performance of \ac{NOMA} without \ac{IRS} assistance. This can be easily obtained by setting $\ell_{\mathrm{BS}} = 0$ in the analysis. \Cref{fig:lowSNR} and \Cref{fig:highSNR} show the outage performance for transmit powers of $20$\,dBm and $35$\,dBm, respectively. This allows us to investigate the performance for both the low and high \ac{SNR} regimes. 

\begin{figure}[t]\label{fig:f3}
	\centering
	\resizebox{0.9\linewidth}{!}{%
		\pgfplotsset{width=260pt, height=195pt,compat = 1.9}
		\begin{tikzpicture}
\begin{semilogyaxis}[
xlabel={Outage threshold [dB]},
ylabel={Outage probability},
label style={font=\labelfont},
ylabel shift = -1mm,	
ymin=0.0001, ymax=1,
xmin=-15, xmax=25,
xtick={-15, -10, -5, 0, 5, 10, 15, 20, 25},
ticklabel style = {font=\tickfont},
ymajorgrids=true,
xmajorgrids=true,
yminorgrids=true,
xminorgrids=true,
major x grid style={solid, cgrid},
major y grid style={solid, cgrid},
minor x grid style={solid, cgrid},
minor y grid style={dotted, cgrid},
legend style={font=\legendfont, name=legendNode, at={(0.99, 0.256)}},
legend cell align=left,
%legend pos=south east,
]
\conf{{north west}}{0.025}{0.97};
% Solid
\foreach \i/\c/\v in {1/greenC1/1, 5/redC1/2, 9/blueC1/3}{
	\edef\temp{\noexpand \addplot [linew, color=\c, mark=none, forget plot] table [x=x\i, y=y\i, col sep=comma] {graphics/results/Fig3.csv};}
	\temp
};
% Dashed
\foreach \i/\c/\v in {2/greenC1/1, 6/redC1/2, 10/blueC1/3}{
	\edef\temp{\noexpand \addplot [linew, color=\c, mark=none, dashed, dash pattern=on 5pt off 5pt, forget plot] table [x=x\i, y=y\i, col sep=comma] {graphics/results/Fig3.csv};}
	\temp
};
% Error bars
\foreach \i/\c/\m/\v in {3/greenC1/square*/1, 4/greenC1/square*/1, 7/redC1/*/2, 8/redC1/*/2}{
	\edef\temp{\noexpand \addplot [error bars/.cd, y dir=both, y explicit, error mark options={errormarkSty}] [linew, mark=\m, only marks, marksz, color=\c, forget plot] table [x=x\i, y=y\i, y error minus=cl\i, y error plus=ch\i, col sep=comma] {graphics/results/Fig3.csv};}
	\temp
};
\foreach \i/\c/\m/\v in {11/blueC1/triangle*/3, 12/blueC1/triangle*/3}{
	\edef\temp{\noexpand \addplot [error bars/.cd, y dir=both, y explicit, error mark options={errormarkSty}] [linew, mark=\m, only marks, marksz2, color=\c, forget plot] table [x=x\i, y=y\i, y error minus=cl\i, y error plus=ch\i, col sep=comma] {graphics/results/Fig3.csv};}
	\temp
};
\addlegendimage{short Legend1, color=NOMA, line width=2pt};
\addlegendentry{NOMA only \hspace*{-4pt}};
\addlegendimage{short Legend2, color=NOMAIRS1, line width=2pt};
\addlegendentry{IRS-NOMA boost UE1 \hspace*{-4pt}};
\addlegendimage{short Legend3, color=NOMAIRS2, line width=2pt};
\addlegendentry{IRS-NOMA boost UE2 \hspace*{-4pt}};
\addplot[mark=none, color=black, line width=0.75pt] coordinates {(0,0) (0,1)};
\label{pgfr1}
\addplot[mark=none, color=black, dashed, dash pattern=on 5pt off 5pt, line width=0.75pt] coordinates {(0,0) (0,1)};
\label{pgfr2}
\addplot[mark=square, only marks, color=black, mark size=2pt] coordinates {(0,0) (0,2)};
\label{pgfr3}
\addplot[mark=o, only marks, color=black, mark size=2pt] coordinates {(0,0) (0,2)};
\label{pgfr4}
\addplot[mark=triangle, only marks, color=black, mark size=2.2pt, mark options={rotate=180}] coordinates {(0,0) (0,2)};
\label{pgfr5}
\end{semilogyaxis}
\node [draw, fill=white, above=2pt of legendNode.north east, anchor=south east](n1) {\shortstack[l]{
		\legendfont \ref*{pgfr1}  UE1, analysis \\    
		\legendfont \ref*{pgfr2}  UE2, analysis \\  
		\legendfont \,\ref*{pgfr3}~\ref*{pgfr4}~\ref*{pgfr5}~\,Simulation
}};
\end{tikzpicture}
	}
	\vspace{-2mm}
	\caption{Outage performance at low \ac{SNR} ($P = 20$\,dBm).}
	\label{fig:lowSNR}
	\vspace{-3mm}
\end{figure}
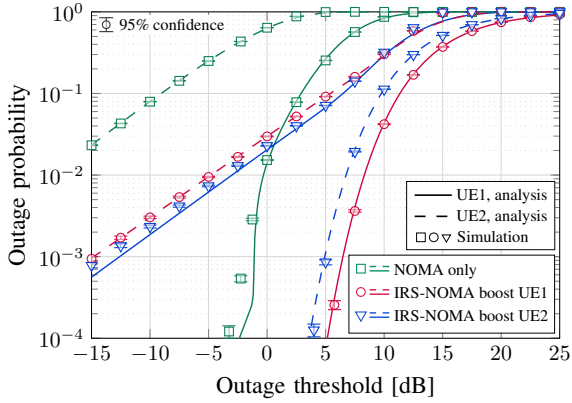

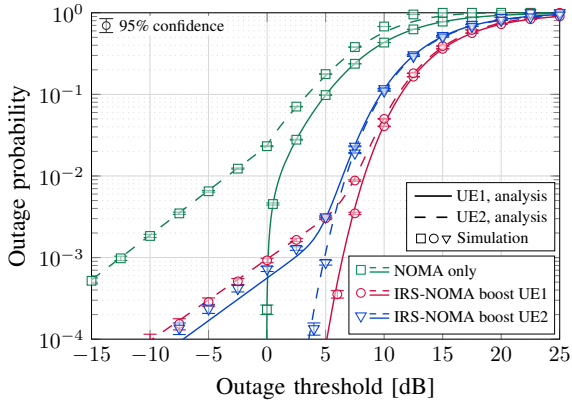
\begin{figure}[t]\label{fig:f4}
	\centering
	\resizebox{0.9\linewidth}{!}{%
		\pgfplotsset{width=260pt, height=195pt,compat = 1.9}
		\begin{tikzpicture}
\begin{semilogyaxis}[
xlabel={Outage threshold [dB]},
ylabel={Outage probability},
label style={font=\labelfont},
ylabel shift = -1mm,	
ymin=0.0001, ymax=1,
xmin=-15, xmax=25,
xtick={-15, -10, -5, 0, 5, 10, 15, 20, 25},
ticklabel style = {font=\tickfont},
ymajorgrids=true,
xmajorgrids=true,
yminorgrids=true,
xminorgrids=true,
major x grid style={solid, cgrid},
major y grid style={solid, cgrid},
minor x grid style={solid, cgrid},
minor y grid style={dotted, cgrid},
legend style={font=\legendfont, name=legendNode, at={(0.99, 0.256)}},
legend cell align=left,
%legend pos=south east,
]
\conf{{north west}}{0.025}{0.97};
% Solid
\foreach \i/\c/\v in {1/greenC1/1, 5/redC1/2, 9/blueC1/3}{
	\edef\temp{\noexpand \addplot [linew, color=\c, mark=none, forget plot] table [x=x\i, y=y\i, col sep=comma] {graphics/results/Fig4.csv};}
	\temp
};
% Dashed
\foreach \i/\c/\v in {2/greenC1/1, 6/redC1/2, 10/blueC1/3}{
	\edef\temp{\noexpand \addplot [linew, color=\c, mark=none, dashed, dash pattern=on 5pt off 5pt, forget plot] table [x=x\i, y=y\i, col sep=comma] {graphics/results/Fig4.csv};}
	\temp
};
% Error bars
\foreach \i/\c/\m/\v in {3/greenC1/square*/1, 4/greenC1/square*/1, 7/redC1/*/2, 8/redC1/*/2}{
	\edef\temp{\noexpand \addplot [error bars/.cd, y dir=both, y explicit, error mark options={errormarkSty}] [linew, mark=\m, only marks, marksz, color=\c, forget plot] table [x=x\i, y=y\i, y error minus=cl\i, y error plus=ch\i, col sep=comma] {graphics/results/Fig4.csv};}
	\temp
};
\foreach \i/\c/\m/\v in {11/blueC1/triangle*/3, 12/blueC1/triangle*/3}{
	\edef\temp{\noexpand \addplot [error bars/.cd, y dir=both, y explicit, error mark options={errormarkSty}] [linew, mark=\m, only marks, marksz2, color=\c, forget plot] table [x=x\i, y=y\i, y error minus=cl\i, y error plus=ch\i, col sep=comma] {graphics/results/Fig4.csv};}
	\temp
};
\addlegendimage{short Legend1, color=NOMA,line width=2pt};
\addlegendentry{NOMA only \hspace*{-4pt}};
\addlegendimage{short Legend2, color=NOMAIRS1,line width=2pt};
\addlegendentry{IRS-NOMA boost UE1 \hspace*{-4pt}};
\addlegendimage{short Legend3, color=NOMAIRS2,line width=2pt};
\addlegendentry{IRS-NOMA boost UE2 \hspace*{-4pt}};
\end{semilogyaxis}
\node [draw, fill=white, above=2pt of legendNode.north east, anchor=south east](n1) {\shortstack[l]{
		\legendfont \ref*{pgfr1}  UE1, analysis \\    
		\legendfont \ref*{pgfr2}  UE2, analysis \\  
		\legendfont \,\ref*{pgfr3}~\ref*{pgfr4}~\ref*{pgfr5}~\,Simulation
}};
\end{tikzpicture}
	}
	\vspace{-2mm}
	\caption{Outage performance at high \ac{SNR} ($P = 35$\,dBm).}
	\label{fig:highSNR}
	\vspace{-6mm}
\end{figure}

The first observation we make is that the weak \ac{NLOS}-dominated user (\ac{UE}2) always benefits from the deployment of the \ac{IRS}, whether the \ac{IRS} is configured to boost its power, or the other \ac{UE}. We see this trend at both low and high \ac{SNR}. For the other \ac{LOS}-dominated user (\ac{UE}1), the story is different. At low outage thresholds, and when the \ac{IRS} is configured to boost \ac{UE}2, the  performance of \ac{UE}1 is worse compared to the \ac{NOMA} only case, especially at low \ac{SNR}. This can be explained as follows. When the \ac{IRS} is configured to boost \ac{UE}2, the combining will appear random for \ac{UE}1, i.e., the fading amplitude will consist of a dominant \ac{LOS} plus a group of randomly combined amplitudes. Due to those additional added components, the \ac{LOS}-dominated fading is destroyed, and the probability that a \ac{UE} goes into a deeper fade becomes higher, and this is reflected in its outage probability. \begin{reply}We can infer this degradation from the expressions in our analysis; at low outage thresholds, the performance is dominated by noise, i.e., $	p_{\mathrm{out,\,IC}}^{(i)} \approx p_{\mathrm{out,\,SNR}}^{(i)}$, meaning that we only need to consider the statistics of $Z_2$. If we assume that only the direct path exists and it has a strong LOS ($m_{h_2} \rightarrow \infty$), then its moments are given by $\mu_{Z_2}^{\vphantom{(2)}} = \ell_{h_2}$ and $\mu_{Z_2}^{(2)} = \ell_{h_2}^2$. If we now perform the moments matching, then this results in a scale parameter $\theta_2 \rightarrow 0$, leading to a very fast decaying tail of the Gamma distribution, which corresponds to a steep slope of the outage curve. If we now add the reflection path, then the difference between $\mu_{Z_2}^{\vphantom{(2)}}$ and $\mu_{Z_2}^{(2)}$ will increase, leading to a higher $\theta_2$, i.e., more spread of the distribution. Since the terms are out-of-phase under random combining, then the spread will occur in both directions around the \ac{LOS} component and results in a wider tail, leading to a flatter outage slope. This indicates that it might not be a good idea to pair a \ac{UE} that has a strong \ac{LOS} to both the \ac{BS} and the \ac{IRS} with another \ac{UE}, while configuring the \ac{IRS} to boost that other \ac{UE}. Nevertheless, one can argue that for the \ac{LOS} \ac{UE}, it is received with high power at the \ac{BS}, and therefore it will probably try to operate at high rates, corresponding to high outage thresholds. In that case, and as can be seen in \Cref{fig:highSNR}, the outage performance of \ac{UE}1 is always better at sufficiently high outage thresholds. Alternatively, instead of configuring the entire \ac{IRS} to boost either of the \acp{UE}, one possibility is to split the surface between the two \acp{UE}, such that a certain performance is guaranteed for each \acp{UE}. This has the potential of reducing the degradation due to the \ac{LOS} conditions mentioned above.\end{reply}

%Final remarks: the expressions described in this paper contain division between Gamma functions. This can be problematic stability-wise when implemented. We suggest performing the calculations in the $\log$-domain using the $\log \Gamma(.)$ function. \\
%\textbf{Research reproducibility:} the code for generating the results in this paper can be downloaded here *TBD*.

\vspace{-2mm}
\section{Conclusions}
In this letter, we investigate the outage performance of an \ac{IRS}-assisted \ac{NOMA} uplink, where all users have direct and reflection links, and all links undergo Nakagami-$m$ fading. Using second-order moments matching, the received powers of the \ac{NOMA} users are approximated with Gamma \acp{RV}. This allows for flexible modeling of the propagation environment, while giving rise to tractable outage expressions under \ac{IC}. We analyzed an example scenario in which one of the \acp{UE} has a dominant \ac{LOS} connection to the \ac{BS}, while the other \ac{UE} has a dominant \ac{NLOS}, and made the following observations:
\begin{itemize}
	\item Presence of the \ac{IRS} always improves the performance of the \ac{NLOS} \ac{UE}, irrespectively of the \ac{IRS} configuration.
	\item At low outage thresholds, the presence of the \ac{IRS} might degrade the performance of the \ac{LOS}-dominated \ac{UE}, if the \ac{IRS} is configured to boost the other \ac{UE}. This is especially pronounced at low \ac{SNR}.
\end{itemize}
The accuracy of the analysis is verified by simulations.

\vspace*{-3mm}
\begin{acronym}[DSTTDSGRC]
\setlength{\itemsep}{-3pt}
\acro{CS}{compressed sensing}
\acro{ETF}{equiangular tight frame}
\acro{OGF}{orthoplectic Grassmannian frame}
\acro{NOMA}{non-orthogonal multiple access}
\acro{OMA}{orthogonal multiple access}
\acro{DFT}{discrete Fourier transform}
\acro{CDMA}{code-division multiple-access}
\acro{BCASC}{best complex antipodal spherical codes}
\acro{CBGC}{coherence-based Grassmannian codebook}
\acro{ICBP}{iterative collision-based packing}
\acro{MMSE}{minimum mean squared error}
\acro{MUSA}{multi-user shared access}
\acro{SIC}{successive interference cancellation}
\acro{SNR}{signal-to-noise ratio}
\acro{TDL-C}{tapped-delay-line-C}
\acro{LTE}{long-term evolution}
\acro{SINR}{signal-to-interference-plus-noise ratio}
\acro{SVD}{singular value decomposition}
\acro{KKT}{Karush-Kuhn-Tucker}
\acro{BLER}{block error ratio}
\acro{5G}{fifth-generation}
\acro{6G}{sixth-generation}
\acro{B5G}{beyond fifth-generation}
\acro{IoT}{internet-of-things}
\acro{PAPR}{peak-to-average-power ratio}
\acro{FFT}{fast-Fourier-transform}
\acro{IFFT}{inverse fast-Fourier-transform}
\acro{OFDM}{orthogonal frequency-division multiplexing}
\acro{BS}{base station}
\acro{UE}{user equipment}
\acro{MUD}{multiuser detection}
\acro{CWL}{codeword level}
\acro{MMSE}{minimum mean square error}
\acro{MF}{matched filter}
\acro{PIC}{parallel interference cancellation}
\acro{CRC}{cyclic-redundancy-check}
\acro{RB}{resource-block}
\acrodefplural{RB}{resource-blocks}
\acro{RMS}{root-mean-square}
\acro{DS}{delay spread}
\acro{LDPC}{low-density parity-check}
\acro{MIMO}{multiple-input multiple-output}
\acro{ITS}{intelligent transport systems}
\acro{V2X}{vehicle-to-everything}
\acro{V2V}{vehicle-to-vehicle}
\acro{V2I}{vehicle-to-infrastructure}
\acro{V2N}{vehicle-to-network}
\acro{V2P}{vehicle-to-pedestrian}
\acro{DSRC}{dedicated short-range communication}
\acro{C-V2X}{cellular-\ac{V2X}}
\acro{IEEE}{institute of electrical and electronics engineers}
\acro{MAC}{medium access control}
\acro{PHY}{physical}
\acro{CSMA}{carrier sense multiple access}
\acro{SB-SPS}{sensing-based semi-persistent scheduling}
\acro{5G-NR}{5th generation new-radio}
\acro{mMTC}{massive machine-type communication}
\acro{IGMA}{interleave-grid multiple access}
\acro{IDMA}{interleave-division multiple access}
\acro{ECDF}{empirical cumulative distribution function}
\acro{LLR}{log-likelhood-ratio}
\acro{IRS}{intelligent reflecting surface}
\acro{RIS}{reconfigurable intelligent surface}
\acro{LOS}{line-of-sight}
\acro{NLOS}{non-line-of-sight}
\acro{RV}{random variable}
\acro{CLT}{central limit theorem}
\acro{CDF}{cumulative distribution function}
\acro{IC}{interference cancellation}
\end{acronym}
\bibliographystyle{IEEEtran}
\bibliography{IEEEabrv,./References,./ExternalReferences,./LocalRefs}

\end{document}